\documentclass[letterpaper, 10 pt, conference]{ieeeconf}  

\IEEEoverridecommandlockouts                              

\usepackage{graphics} 
\usepackage{epsfig} 
\usepackage{times} 
\usepackage{amsmath} 
\usepackage{amssymb}  
\usepackage{xcolor}
\usepackage{caption}
\usepackage{subcaption}
\usepackage{algorithmic}
\usepackage{algorithm}
\usepackage[english]{babel}
\usepackage{semantic}
\usepackage{soul}
\usepackage{mathtools}
\usepackage[hidelinks]{hyperref}

\newtheorem{definition}{Definition}
\newtheorem{theorem}{Theorem}
\newtheorem{lemma}{Lemma}
\newtheorem{remark}{Remark}

\title{\LARGE \bf
Distributed State Estimation for Linear Time-Varying Systems with Sensor Network Delays
}

\author{Sanjay Chandrasekaran*, Vishnu Varadan*, Siva Vignesh Krishnan, Florian Dörfler, Mohammad H. Mamduhi
\thanks{S. Chandrasekaran, V. Varadan, F. Dörfler and M. H. Mamduhi are with the Automatic Control Laboratory, ETH Z\"urich. S. V. Krishnan is with the Department of Mechanical \& Process Engineering, ETH Z\"urich.}\thanks{
* These authors contributed equally to this work.}\thanks{
E-mail of the corresponding author: schandraseka@ethz.ch
}}

\begin{document}
\maketitle
\thispagestyle{empty}
\begin{abstract}
Distributed sensor networks often include a multitude of sensors, each measuring parts of a process state space or observing the operations of a system. Communication of measurements between the sensor nodes and estimator(s) cannot realistically be considered delay-free due to communication errors and transmission latency in the channels. 
    We propose a novel stability-based method that mitigates the influence of sensor network delays in distributed state estimation for linear time-varying systems. 
    Our proposed algorithm efficiently selects a subset of sensors from the entire sensor nodes in the network based on the desired stability margins of the distributed Kalman filter estimates, after which, the state estimates are computed only using the measurements of the selected sensors.
    We provide comparisons between the estimation performance of our proposed algorithm and a greedy algorithm that exhaustively selects an optimal subset of nodes. 
     We then apply our method to a simulative scenario for estimating the states of a linear time-varying system using a sensor network including 2000 sensor nodes. 
    Simulation results demonstrate the performance efficiency of our algorithm and show that it closely follows the achieved performance by the optimal greedy search algorithm.
\end{abstract}

\section{Introduction}

Distributed estimation and filtering~\cite{Olfati1} has been extensively studied and widely used to produce accurate estimates of states of dynamical systems in a decentralized fashion. Numerous works have shown that the distributed estimation schemes that compute state estimates based on the measurements of multiple sensor nodes (a sensor network) are considerably more reliable, scalable, and robust compared to the centralized approaches~\cite{speranzon, yang}. All these features result in distributed estimation to be scalable concerning the size of the system being sensed and the size of the sensor network~\cite{talebi}, hence being an applicable method to sense and monitor large-scale and networked systems using dense sensor networks.

\begin{figure}
    \centering
    \includegraphics[scale=0.48]{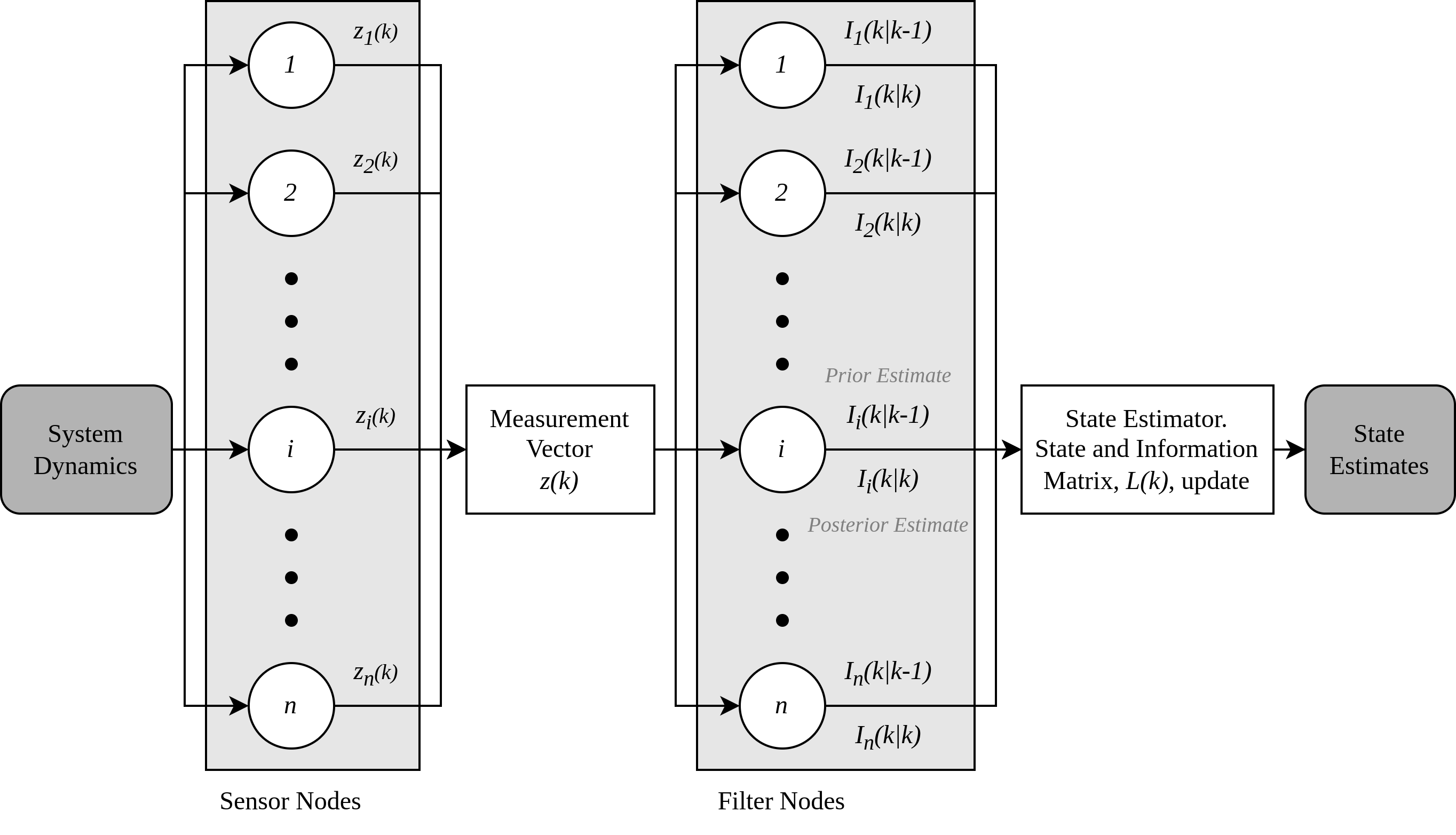}
     \caption{Representative DKF block diagram redrawn from \cite{MesbahiEgerstedt_2010}.
     }
    \label{fig:fig1}
    \vspace{-6 mm}
\end{figure}

Estimation techniques using sensor networks are typically classified into centralised and decentralised categories. In the former, sensor nodes directly communicate with a centralized ``fusion'' centre, which computes state estimates. While centralized methods may require fewer computations, it is not very dependable as a single node failure could result in loss of observability~\cite{speranzon}. In addition, centralized methods suffer from scalability issues \cite{hao}. To tackle these shortcomings, decentralized and distributed estimation and filtering algorithms have been proposed~\cite{Olfati1,talebi,MesbahiEgerstedt_2010}. In Distributed Kalman Filtering (DKF), the sensor nodes mutually exchange information with their neighbours and produce state estimates as well as estimation errors. Such information fusion technique makes the estimation process more reliable in the sense that partial node failure leads only to local estimation inaccuracies. DKF architectures make the estimation process more dependable and robust by eliminating a single fusion centre. However, these approaches are affected by the burden of communication overload and propagation of communication delays, especially in dense sensor networks~\cite{hao}.

 Standard DKF algorithms, such as in~\cite{MesbahiEgerstedt_2010} consider an ideal scenario, wherein delay-free measurements are collected from a network of multiple sensor nodes, after which an estimator aggregates all state measurements and covariance matrix estimates to compute the state estimates. A consensus-based fusion of sensor data was proposed in \cite{Olfati1}, where distributed micro-filters were associated with each sensor node to run local estimation processes. The consensus-based filtering was introduced to facilitate efficient fusion of measurements.
Another consensus-based fusion technique was proposed in~\cite{BATTISTELLI}, implementing event-triggered communication between sensor nodes. In this work, the Kalman Filter (KF) algorithm is run on each sensor node, and to improve the estimate of a node, event-triggered strategy is used to trigger information exchange between the nodes and their neighbours.
{However, the presence of packet losses, noise in the transmission channels, and bandwidth limitations~\cite{liang,mamduhi1} lead to the existence of delays from individual sensor nodes to the estimator ~\cite{dimitris}.} Delays may result in estimation inaccuracies, lowering control performance, and even destabilization of the system. Hence, it is of importance to understand the effects of the limitations associated with communication between the sensors and the estimator.
 Intermittent arrival of sensor measurements has been studied in~\cite{mamduhi2} where modifications in the prediction part of the KF algorithm are made to mitigate the absence of data packets. Along similar lines, in \cite{liang}, a time-varying Bernoulli-distributed random variable was defined, which would return a value of one if the dynamics at a particular time was affected by the past state and zero if not.
 In \cite{hao1}, the influence of transmission delays and packet losses is mitigated by minimizing the error covariance of the estimated and true state values. This method also accounted for fading measurements.
 However, the technique is developed for Linear Time-Invariant (LTI) systems.

In this work, 
we propose a method to find a subset of nodes from the total set of sensor nodes and use only the measurements of that subset for estimating the states of the system. Studies on subset selection exist in the literature, e.g., in \cite{paola}, wherein a node selection scheme is proposed based on maximizing a ``perception confidence value'' (inverse of the trace of the posterior covariance matrix). 
In this paper, we initially apply a greedy sensor selection algorithm by choosing appropriate thresholds on sensor characteristics such as variance and delays, and use this method for performance comparisons as it is an exhaustive search method that needs to have access to true state values, hence not practical.  We then propose a practical approach for selecting a subset of sensors based on stability criterion for distributed state estimation for LTV systems. We show that the proposed distributed algorithm yields comparable performance with that of the greedy one. The major contributions of this paper are as follows:
\begin{enumerate}
    \item Investigation of distributed state estimation for LTV systems in the presence of delays and a central selection of nodes which would then be used in the DKF. 
    \item Theoretical stability analysis of DKF for LTV systems, by extending the  {stability} results presented in \cite{BATTISTELLI2014707}  {for LTI systems}. 
    \item   {Proposal of} a scalable subset selection for sensor nodes based on  {the above} stability criterion and performance comparison with a greedy algorithm as a benchmark.
\end{enumerate}
The rest of this paper is structured as follows: Section \ref{cha:materials} introduces the problem and provides a brief  overview of DKF over sensor networks. Section \ref{cha:materials1} discusses the influence of delays on the DKF estimates, and the overall system's structural observability. It also provides stability analysis for DKF estimates for LTV systems. Section \ref{cha:materials2} presents the sensor selection algorithms; the conventional greedy method and the proposed method based on the stability criterion. Section \ref{cha:results} presents the simulation results on a second-order LTV system and performance comparisons.

\noindent{\textbf{Notations:}}
 In this paper, $\mathcal{N}(J,K)$ refers to a multi-variate Gaussian distribution with mean vector $J$ and covariance matrix $K$. The notations $\hat{z}(k|k-1)$ and $\hat{z}(k|k)$ represent the \textit{priori} and \textit{posteriori} estimates of $z$, respectively, while $\hat{z}$ indicates an estimate of variable $z$. The expectation operator is represented by $\mathbb{E}\{.\}$. The operator $tr(.)$ represents the trace of a square matrix. The notation $[n]$ denotes the set $\{1,2,\dots,n\}$. The notations $A \succeq B$ ($A \preceq B$) and $A \succ B$ ($A \prec B$) indicate that the matrix $A-B$ is positive semi-definite (negative semi-definite) and positive definite (negative definite), respectively  {and $\mathbb{Z}_{\geq 0}$ refers to the set of all non-negative integers}.

\section{Problem Description and Preliminaries}\label{cha:materials}
\subsection{Problem Setup}
We consider the state estimation problem for an LTV system being observed with a  {sensor network with $n$ sensors to measure the states of the system. Fig. \ref{fig:fig1} depicts the block diagram of the DKF. Distributed estimation is carried out at the "filter nodes", wherein the individual prior and posterior estimates are computed, after which they are aggregated at the "state estimator" to} produce state estimates of the LTV system. While sending information to the state estimator, each sensor is associated with a  {stochastic} communication delay. Delays typically compromise the accuracy of state estimates, hence, the goal is to select the optimal number of sensors that produce accurate-enough estimates in the presence of delays. 
We consider a stochastic LTV system represented by the following discrete-time model
\begin{equation} \label{eq1}
    x(k+1) = A(k)x(k) + w(k), 
\end{equation}
where $x(k) \!\in\! \mathbb{R}^m$ and $w(k) \!\sim\! \mathcal{N}(0_{m},Q)$ denote the  {states and the} system's exogenous disturbances,  {respectively}.  {It is assumed} that the additive disturbances are independent $\forall \  k\in \mathbb{Z}_{\geq 0}$. 
 The states are assumed to be partially observable by individual sensors, with  {the partial measurements} given by
\begin{equation*} \label{eq2}
    z_i(k) = H_ix(k) + v_i(k),
\end{equation*}
where $z_i(k)\in \mathbb{R}^{p}$ is the measurement vector taken from the $i$-th sensor node at time-step $k$, $H_i \in \mathbb{R}^{p \times m}$, and $v_i \sim \mathcal{N}(0_{p},R_i)$ is the measurement noise associated with the sensor $i$. The aggregate measurement, i.e., including measurements from all sensors, would then be expressed as
\begin{align*} 
z(k) &= Hx(k) + v(k),\\
z(k) = [z_1(k) \;z_2(k) & \;\ldots\;z_n(k)]^\top,  H = [H_1 \;H_2 \;\ldots \;H_n]^\top, \\v(k) &= [v_1(k) \;v_2(k)\; \ldots\; v_n(k)]^\top.
\end{align*}

\vspace{-5.5mm}
\subsection{Preliminaries}

 {The conventional KF state update equation is given by:}
\begin{equation*} \label{eq4}
    \hat{x}(k|k) = \hat{x}(k|k-1) + L(k)(z(k) - H\hat{x}(k|k-1)),
\end{equation*}
where $L(k)$ is the time-varying Observer Gain (OG). 
We define the posteriori estimation error as $\epsilon(k|k) = \hat{x}(k|k) - x(k)$, and its associated posteriori covariance matrix by $\Sigma(k|k)=\mathbb{E}\{\epsilon(k|k) \epsilon(k|k)^\top\}$.  {By minimizing the} trace of the covariance matrix, the recursive equations for the OG and the covariance matrix are obtained as follows
\begin{align} \label{eq5}
    L(k) &= \Sigma(k|k)H^{\top}R^{-1},\\ \label{eq6}
    \Sigma(k|k) &= (I - L(k)H)\Sigma(k|k-1),
\end{align}
where $R = [R_1,R_2,\ldots,R_n]^\top$. 

For the  {purpose of} brevity, we represent the KF in a compact form \cite{MesbahiEgerstedt_2010} using the Information Filter (IF) framework, wherein we define $\mathcal{I}(k|k) = \Sigma(k|k)^{-1}$, and $\mathcal{I}(k|k-1) = \Sigma(k|k-1)^{-1}$ as the corresponding Information Matrices (IM) for the priori and posteriori state estimates, respectively. Using the IM, the expression \eqref{eq6} together with the OG in \eqref{eq5} can be written as

\begin{equation} \label{eq9}
    \mathcal{I}(k|k) = \mathcal{I}(k|k-1) + Y(k),
\end{equation}
where, $Y(k) = H^\top R^{-1}H$. We further introduce variables $y,\hat{y}$ as Information Vectors (IV), such that $y(k) \!=\! H^\top R^{-1}z(k)$, and $\hat{y}(k|k) = \mathcal{I}(k|k)\hat{x}(k|k)$. We then have
\begin{equation} \label{eq10}
    \hat{y}(k|k) = \hat{y}(k|k-1) + y(k).
\end{equation}
The recursive expressions for the OG and IM, in the IF framework, can therefore be  derived as

\begin{align} \label{eq11}
    L(k) &= A(k)\;\mathcal{I}(k|k)H^{\top}R^{-1},\\\label{eq12}
    \mathcal{I}(k\!+\!1|k) &=\! (I \!-\! C(k))^{\!\top}\!M(k)(I \!-\! C(k)) \!+ C(k)Q^{-1}C(k)^{\!\top}\!,\\\label{eq13}
    \hat{y}(k\!+\!1|k) &= (I - C(k))(A(k)^{-1})^{\top}\hat{y}(k|k),
\end{align}
where, $M(k) = (A(k)^{-1})^{\top}\mathcal{I}(k|k)A(k)^{-1}$, $C(k) = M(k)(M(k) + Q^{-1})^{-1}$.  {In case $A(k)$ gets close to singularity, its inverse can be replaced by the pseudo-inverse \cite{MesbahiEgerstedt_2010}.}

\subsection{Distributed Kalman Filter}

 At the  {state} estimator, the difference $\mathcal{I}_i(k|k) - \mathcal{I}_i(k|k-1)$ is computed and the IM is then updated. Similarly, the corresponding IV is updated by computing $\hat{y}_i(k|k) - \hat{y}_i(k|k-1)$  {$\forall i \in [n]$}. Each filter node provides its corresponding prior and posterior IM, i.e., $\mathcal{I}_i(k|k)$ and $\mathcal{I}_i(k|k-1)$, and state estimates, i.e., $\hat{x}_i(k|k)$ and $\hat{x}_i(k|k-1)$. This information is forwarded to the system estimator that produces the following outputs for the overall set of state estimates:
\begin{align} \label{eq16}
    \hat{x}(k|k) &= \Sigma(k|k)[\mathcal{I}(k|k-1)\hat{x}(k|k-1)
    \\\nonumber &+\sum_{j=1}^{n}{[\mathcal{I}_j(k|k)\hat{x}_j(k|k) - \mathcal{I}_j(k|k-1)\hat{x}_j(k|k-1)]},\\ \label{eq17}
    \mathcal{I}(k|k) &= \mathcal{I}(k|k-1) + \sum_{j=1}^{n}{[\mathcal{I}_j(k|k)-\mathcal{I}_j(k|k-1)]}.
\end{align}

\section{Structural Observability and Stability Analysis} \label{cha:materials1}
\subsection{DKF with Delays}

In this section, we present the modifications required for the DKF in the presence of stochastic delays. Before proceeding, we assume that the open-loop stochastic LTV system in \eqref{eq1} is asymptotically stable in the mean-square sense. This requires $Q$ to be finite and the state transition matrix, $\Phi(k,0) \triangleq A(k-1)A(k-2)\dots A(0) \rightarrow 0$, $k\rightarrow{\infty}$ .

If there exists a delay $\tau : N-\tau>0$, where $N$ is some time index; $\hat{x}_i(k|k):k \in \{N-1,\ldots, N-\tau+1\}$ is unavailable from a particular filter node $i$, leading to
\begin{align*}
        \hat{x}_i(N|N-1) &=  [\hat{x}_i(N-\tau+1|N-\tau) + z_i(N-\tau)]\\ &\times[\prod_{j=1}^{\tau-1}{\{A(N-j) - L_i(N-j)H_i(N-j)\}}]
\end{align*}

\subsection{Structural Observability of the Networked System}
In a networked dynamical system with a large number of nodes, the classical tests for observability that include calculating the Grammian matrix would become easily complicated when scale grows. This motivates us to look into the concept of ``structural observability'' to identify the number of nodes to guarantee the system observability. Initially developed in \cite{1100557} for structural controllability, the \textit{theory of duality} enables us to extend this concept to structural observability. Next, we briefly outline the concept of structural observability, mainly relying on the results in \cite{montanari}.

\begin{definition}
A matrix $\bar A \in$ \{0,$\star$\} is said to be a structural matrix if each element of $\bar A = [a_{ij}]$ is either zero or a non-zero free parameter denoted by a $\star$. A matrix $A$ is called a numerical realization of $\bar A$ if any real number is assigned to all the free parameters of $\bar A$.
\end{definition}

The structured pair of matrices ($\bar A, \bar H$) is said to be structurally observable if and only if a numerical realization $(A,H)$ is observable. Structurally observable pair of matrices is observable in a wide range of parameters except in a narrow range that renders the system unobservable. Clearly, structural observability is only a necessary condition for the system to be observable in Kalman's sense and does not prove sufficiency. In \cite{liu} it is proven that the system cannot be structurally controllable, if in the associated state-control input graph (this graph is formed by taking the states and control inputs as the vertices, wherein a directed edge between the two states $x_j$ and $x_i$ exists if $a_{ij}$ is nonzero, and a directed edge between the control input $u_j$ and state  $x_i$ exists if $b_{ij}$ is nonzero), one cannot reach all the states, starting from any of the control input nodes. By the theory of duality, we can extend the above result to define structural observability. The pair $(A,H)$ is said to be structurally observable if and only if, in the respective graph formed by $A^\top$ and $H^\top$, every state $x_i$ can be reached from   {the} outputs $y_j$   {and no dilation exists}. This result also extends to linear time varying systems, as shown in \cite{ltvobservability}.

\subsection{Stability of DKF}\label{sub:stability}
A comprehensive stability analysis of DKF for LTI systems is presented in~\cite{BATTISTELLI2014707}. In this section, we extend that result by conducting a rigorous analysis of the stability for LTV systems using DKF and prove that the  estimation error $x_i(k) - \hat{x}_i(k|k)$ remains bounded in the mean square sense. Before stating the main result, we state the following lemma.
\begin{lemma}
Assume that $A(k)$ is invertible for all $k$. Recalling the recursive form of the IM, one can write from~\eqref{eq12}
\begin{equation*}
\mathcal{I}_i(k+1|k) = \psi_k(\mathcal{I}_i(k|k)).    
\end{equation*}The following results then hold:
\begin{enumerate}
    \item The function $\psi_k(.)$ is monotonic and non-decreasing, i.e., given two positive semi-definite matrices $\mathcal{I}_1$ and $\mathcal{I}_2$, with $\mathcal{I}_1 \preceq \mathcal{I}_2$, we have $0 \preceq \psi_k(\mathcal{I}_1) \preceq \psi_k(\mathcal{I}_2)$.
    \item For any positive semi-definite matrix $\widetilde{\mathcal{I}}$, there exists a $\hat{\beta} \in \left(0,1\right]$ 
    such that $\psi_k(\mathcal{I}) \succeq \hat{\beta}(A(k)^{-1})^{\top}\mathcal{I}A(k)^{-1}$,  $\forall\; \mathcal{I} \preceq \widetilde{\mathcal{I}}$.
\end{enumerate}
\end{lemma}

\begin{proof}We can rewrite expression~\eqref{eq12} as,
\begin{align*}
    \psi_k(\mathcal{I}) &= (A(k)^{-1})^{\top}\mathcal{I}A(k)^{-1} \\&- (A(k)^{-1})^{\top}\mathcal{I}{(\mathcal{I} + A(k)^{\top}Q^{-1}A(k))}^{-1}\mathcal{I}A(k)^{-1}.
\end{align*}
For a non-singular $\mathcal{I}$, above expression can be simplified to
\begin{equation} \label{eq:psik}
    \psi_k(\mathcal{I}) = {(A(k)\mathcal{I}^{-1}A(k)^{\top} + Q)}^{-1}.
\end{equation}
Consider two distinct matrices $\mathcal{I}_1$ and $\mathcal{I}_2$ where $\mathcal{I}_1 \preceq \mathcal{I}_2$. Let $\mathcal{I}_1(\alpha) = \mathcal{I}_1 + \alpha I$ and $\mathcal{I}_2(\alpha) = \mathcal{I}_2 + \alpha I$. From the above simplified form of $\psi_k(.)$, one can conclude
\begin{align*}
    &\psi_k(\mathcal{I}_1(\alpha))^{-1} - \psi_k(\mathcal{I}_2(\alpha))^{-1} =\\ &A(k)(\mathcal{I}_1(\alpha)^{-1} - \mathcal{I}_2(\alpha)^{-1})A(k)^{\top} \succeq 0.
\end{align*}
Thus, we obtain $\psi_k(\mathcal{I}_{2}(\alpha)) - \psi_k(\mathcal{I}_{1}(\alpha)) \succeq 0$, and the first part of the Lemma then follows by letting $\alpha \to 0^{+}$.

For the second part, let $\mathcal{I}(\alpha) = \mathcal{I} + \alpha I$ with $\alpha > 0$. By using the equality (\ref{eq:psik}), we can write
\begin{align*}
   & \psi_k(\mathcal{I}(\alpha)) =\\& (A(k)^{-1})^{\top}(\mathcal{I}(\alpha)^{-1} +  A(k)^{-1}Q(A(k)^{-1})^{\top})^{-1}A(k)^{-1}.
\end{align*}
For any $\hat{\alpha} > \alpha$, and ${\hat{\mathcal{I}} \succeq \mathcal{I}}$ we have 
\begin{equation*}
    \mathcal{I}(\alpha)^{-1} \succeq (\hat{\mathcal{I}} + \hat{\alpha}I)^{-1},
\end{equation*}
which implies that there exists a positive scalar $\hat{\gamma}(k)$ at each time instant $k$ such that, $\forall k$ we have
\begin{equation*}
    A(k)^{-1}Q(A(k)^{-1})^{\top} \preceq \hat{\gamma}(k)\mathcal{I}(\alpha)^{-1}.
\end{equation*}
Combining the last two expressions, we conclude $\forall k$ that
\begin{equation}\label{betahat}
    \psi_k(\mathcal{I}(\alpha)) \succeq \frac{1}{1+\hat{\gamma}(k)}(A(k)^{-1})^{\top}\mathcal{I}(\alpha)A(k)^{-1}.
\end{equation}
The proof of the second part then follows by setting $\hat{\beta} = \min\limits_{k}\frac{1}{1+\hat{\gamma}(k)}$ and $\alpha \to 0^{+}$.
\end{proof}

\begin{theorem}\label{the:theorem1}
Assume that $A(k)$ is invertible  for all $k$ and define 
$\mathcal{I}_i(0| -1) = 0, \forall i$, i.e., $\hat{x}_i(0| -1) = 0, \forall i$. Then,  for each $k = \{0, 1, \ldots\}$, and $\forall i \in [n]$, we have
\begin{align*}
    \mathcal{I}_i(k|k) \preceq  &\; {\mathbb{E}\{(x(k)-\hat{x}_i(k|k)){(x(k)-\hat{x}_i(k|k))}^\top\}}^{-1}.
\end{align*}
\noindent In other words, the IM is upper bounded by the inverse of the error covariance of the DKF estimates.
\end{theorem}

\begin{proof} The proof is based on mathematical induction. Assume at a time $k$, the following inequality holds $\forall i \in [n]$:
\begin{align*}
    \mathcal{I}_i&(k|k-1) \\&\preceq {\mathbb{E}\{(x(k)-\hat{x}_i(k|k-1)){(x(k)-\hat{x}_i(k|k-1))}^{\top}\}}^{-1}.
\end{align*}
Define $l_i = H_i^{\top}R_i^{-1}H_i$. This implies that 
\begin{align*}
    &{\mathbb{E}\{(x(k)-\hat{x}_i(k|k)){(x(k)-\hat{x}_i(k|k))}^{\top}\}}^{-1} \\
    &= {\mathbb{E}\{(x(k)-\hat{x}_i(k|k-1)){(x(k)-\hat{x}_i(k|k-1))}^{\top}\}}^{-1} + l_i 
    \\ &\succeq \mathcal{I}_i(k|k-1) + l_i\\& = \mathcal{I}_i(k|k)
\end{align*}
From the fist part of Lemma 1, one can readily conclude
\begin{align*}
    &\mathcal{I}_i(k+1|k) = \psi_k(\mathcal{I}_i(k|k))\preceq\\
    & \psi_k({\mathbb{E}\{(x(k)-\hat{x}_i(k|k)){(x(k)-\hat{x}_i(k|k))}^{\top}\}}^{-1})=\\ 
    & {\mathbb{E}\{(x(k+1)-\hat{x}_i(k+1|k)){(x(k+1)-\hat{x}_i(k+1|k))}^{\!\top}\}}^{-1}.
\end{align*}
The proof can be concluded since at $k=0$, the initialization $\mathcal{I}_i(0| -1) = 0, \forall i$, i.e., $\hat{x}_i(0| -1) = 0, \forall i$, ensures that the inequality holds for any $k>0$.
\end{proof}

Theorem \ref{the:theorem1} proves that the true error covariance resulting from the DKF Algorithm is upper bounded by the inverse of the corresponding IM (which is the inverse of the predicted error covariance). To prove the boundedness of the true error covariance, it would therefore be sufficient to show that individual IMs are lower bounded by a positive definite matrix. Having this, we now state the main stability result.
\begin{theorem}\label{the:theorem2}
Assume that $A(k)$ is invertible  for all $k$ and that the structured pair of matrices ($\bar A, \bar H$), where $\bar A$ is the structure of the system matrix with $A(k)$ is structurally observable. Then, there exists a time instant $\bar{k}$ and a positive semi-definite matrix $\widetilde{\mathcal{I}}$ such that 
\begin{equation*}
    0 \prec \widetilde{\mathcal{I}}(k) \preceq \mathcal{I}_i(k|k), \; \forall i \in [n], \;\forall k \geq \bar{k}.
\end{equation*}
Consequently, we obtain
\begin{align}
    \lim_{k \to \infty} \sup {\mathbb{E}\{(x(k)-\hat{x}_i(k|k)){(x(k)-\hat{x}_i(k|k))}^{\top}\}}&  \nonumber\\ \leq tr({\widetilde{\mathcal{I}}}(k)^{-1})&.
\end{align}
\end{theorem}

\begin{proof}
From the definition of $l_i$ in proof of Lemma~1, we have $\mathcal{I}_i(k|k) = \mathcal{I}_i(k|k-1) + l_i$. Hence, we can write
\begin{equation*}
    \mathcal{I}_i(k|k) = \psi_{k-1}(\mathcal{I}_i(k-1|k-1)) + l_i.
\end{equation*}
Further from Theorem 1, we have 
\begin{align*}
    \mathcal{I}_i(k|k) &\preceq {\mathbb{E}\{(x(k)-\hat{x}_i(k|k)){(x(k)-\hat{x}_i(k|k))}^{\top}\}}^{-1} \\&\preceq \mathcal{I}(k|k).
\end{align*}
 Recalling that $Q$ and $R_i, i \in [n]$ are positive definite and the sequence $\mathcal{I}(k|k)$ can be uniformly bounded by some constant matrix $\widetilde{\mathcal{I}}$, we can use the second statement of Lemma 1  to obtain the following lower bound
\begin{equation*}
    \mathcal{I}_i(k|k) \succeq \hat{\beta}(A(k-1)^{-1})^{\top}\mathcal{I}_i(k-1|k-1){A(k-1)}^{-1} + l_i.
\end{equation*}
By applying the above two inequalities recursively, we obtain
\begin{align}
    \mathcal{I}&_i(k|k) \nonumber\\ 
    &\succeq \hat{\beta}^{\bar{k}}((\prod_{i=1}^{\bar{k}}A(k-i))^{-1})^{\top}\mathcal{I}_i(k-\bar{k}|k-\bar{k})\prod_{i=1}^{\bar{k}}A(k-i))^{-1} \nonumber\\
    &+ \!\sum_{\tau=1}^{\bar{k}}\hat{\beta}^{\tau-1}((\prod_{i=1}^{\tau-1}\!A(k-i))^{-1})^{\top}l_i(\prod_{i=1}^{\tau-1}\!A(k-i))^{-1}.\!
\end{align}
Define 
\begin{equation}\label{Itildeeqn}
    \widetilde{\mathcal{I}}(k) = \sum_{\tau=1}^{\bar{k}}\hat{\beta}^{\tau-1}((\prod_{i=1}^{\tau-1}A(k-i))^{-1})^{\top}l_i(\prod_{i=1}^{\tau-1}A(k-i))^{-1},
\end{equation}
and the proof is then complete.
\end{proof}
Although the authors prove stability for a consensus based DKF algorithm~in \cite{BATTISTELLI2014707}, the stability result holds regardless of the number of consensus steps. Clearly, from Theorems 1 and 2, one can see that the estimation error is asymptotically bounded in mean square and the DKF algorithm generates estimates with bounded mean square errors for LTV systems.   {Theorems 1 and 2, in their present form can only be applied for systems in which $A(k)$ is invertible  for all $k$. Extension to the cases in which $A(k)$ is not invertible should be investigated further.}

\begin{algorithm}[!t]
\caption{$[nodes,MSE,d_{max}] = Greedy[\tau_{m},R_{m}]$}
	\label{alg:alg1}
	\begin{algorithmic}
		\STATE $nodes \longleftarrow \emptyset$
		\FOR{$k=1$ \TO $iter$}
		    \STATE $R_0(k) \longleftarrow R_{m}(1 - \frac{k-1}{iter})$
		    \STATE $\tau_0(k) \longleftarrow \tau_{m}(1 - \frac{k-1}{iter})$
		    \FOR{$i=1$ \TO $n$}
		        \IF{$R_i \leq R_0(k)$ \& $\tau_i \leq \tau_0(k)$}
		        \STATE $nodes(k) \longleftarrow nodes(k) \cup i$
		        \ENDIF
		        \STATE $[\hat{X},\mathcal{\hat{I}}]  \longleftarrow DKF[nodes(k)]$
		    \ENDFOR
		    \STATE $MSE(k) \longleftarrow \frac{1}{2}||\hat{x}-x||_2^2$
		    \STATE $d_{max}(k) \longleftarrow \frac{\mathrm{max}_i\{|\hat{x}_i-x_i|\}}{\mathrm{max}\{|x_i|\}}$
	    \ENDFOR
	    \RETURN $[nodes,MSE,d_{max}]$
	\end{algorithmic} 
\end{algorithm}

\section{Sensor Selection Algorithms}\label{cha:materials2}
\subsection{Greedy Subset Selection}
Typically a higher number of sensor nodes results in more accurate state estimates in an ideal scenario, however, it is not always true in the presence of delays. Simulations show that (Fig. \ref{fig:comp}) in the presence of delays, the accuracy of estimates may decrease when the number of sensor nodes increases. To mitigate the effects of delay, we propose to select a subset of sensor nodes such that the compromising effect of delays is reduced by lowering the number of sensors used for estimation, while the selected sensors render a satisfactory estimation accuracy. 
We measure the influence of delays on the state estimates by determining the relative Maximum Deviations (MD), defined as follows:
\begin{equation} \label{eqmaxdev}
    d_{max} = \frac{\underset{i}{\mathrm{max}}\{|\hat{x_i}-x_i|\}}{\underset{i}{\mathrm{max}}\{|x_i|\}}.
\end{equation} 
We measure the influence of the measurement noise using the Mean Squared Error (MSE) between the ideal and observed estimates, $x$ and $\hat{x}$, respectively. We compute the MSE from the time where the state variables settle to $1\%$ of their final values.  {Intuitively, the influence of measurement noise can be best noticed (independent of the influence of delays) when the system settles down to its equilibrium point.} 


In Algorithm~\ref{alg:alg1}, we perform a greedy search for the best subset of nodes such that the influence of delays and measurement noise is reduced while maintaining a desired state estimate accuracy level. In each iteration, the subset selection algorithm includes all those sensor nodes with possessing delays and measurement noise variances beneath their corresponding thresholds. The search iterations start from the highest possible values for $R$ and $\tau$. In each iteration, all the possible sensor nodes with lower measurement noise variances and delays are selected, based on those the DKF is implemented. 
The values of $R_m$ and $\tau_m$ in Algorithm~\ref{alg:alg1} correspond to the largest measurement noise covariance and delay, respectively.
\subsection{Subset Selection based on Stability Criterion}
We now propose a novel subset selection method  to choose a subset of sensor nodes from the available nodes to ensure the estimation error covariance matrix remains bounded asymptotically, i.e., when $k\to\infty$. The advantage of this subset selection method is that, unlike the greedy algorithm, there is no need to be aware of the true state values, while the selected subset of nodes ensures the stability of the DKF in the presence of delays and noises in the collected measurements. The pseudo-code for subset selection based on stability is shown in Algorithm \ref{alg:alg3}.

\begin{algorithm}[!t]
\caption{$[nodes] = StabilityCheck[n,\bar{k},N,\mathcal{T},T_s]$}
	\label{alg:alg3}
	\begin{algorithmic}
		\STATE $nodes \longleftarrow \emptyset$
		\FOR{$i=1$ \TO $n$}
		    \STATE $ct_{exp} \longleftarrow 0, \ ct_{act} \longleftarrow 0$
		    \FOR{$k=\bar{k}+1$ \TO $N$}
		        \STATE Calculate $\widetilde{\mathcal{I}}(k)$ from (\ref{Itildeeqn})
		        \IF{ $k - \frac{\tau_i}{T_s} > 0$ }
		            \STATE $ct_{exp} \longleftarrow ct_{exp} + 1$
		           \IF{$tr[{I}_i(k-\frac{\tau_i}{T_s}|k-\frac{\tau_i}{T_s})] > tr[\widetilde{\mathcal{I}}(k)]$}
		            \STATE $ct_{act} \longleftarrow ct_{act} + 1$
		           \ENDIF
		      \ENDIF
		    \ENDFOR
		    \IF{$ct_{exp} = ct_{act}$}
		        \STATE $nodes \longleftarrow nodes \cup i$
		    \ENDIF
	    \ENDFOR
	    \RETURN $nodes$
	\end{algorithmic}
\end{algorithm}

In Algorithm~\ref{alg:alg3}, $n$ is the total number of sensor nodes, $N$ is the total number of time samples, $\bar{k} \in [0,N]$ is the time instance from which the stability bound in Theorem \ref{the:theorem2} is satisfied,   {$\mathcal{T} = [\tau_1 \quad \tau_2 \quad \dots \quad \tau_n]^\top \in \mathbb{R}^{n}$} is a vector that contains the delays associated with the nodes, and $T_s \in \mathbb{R}^{+}$ is the sampling time of the filter. 
The essence of the algorithm is to calculate the value of the bounding matrix $\widetilde{\mathcal{I}}(k)$ and pick nodes $i$ such that the delayed IM associated with that node, i.e., ${I}_i(k-\frac{\tau_i}{T_s}|k-\frac{\tau_i}{T_s})$, is larger than the bound described in (\ref{Itildeeqn}), $\forall k$. The algorithm includes only these nodes in the final calculation of estimates by the estimator.

Algorithm~\ref{alg:alg3} always returns a set of nodes that satisfies the stability criterion in (\ref{Itildeeqn}), $\forall k > \bar{k}, k \leq N$   {and thus the order of the evaluation of the nodes by the algorithm does not affect its output}. The algorithm fails to return a non-empty set of nodes only when none of the nodes satisfies the stability criterion for the specified time horizon.  {This occurs when the delays are larger than the estimation horizon $N$.}

\begin{remark}
 {The computational complexity} of calculating $\widetilde{\mathcal{I}}(k)$ in Algorithm~\ref{alg:alg3} is $\mathcal{O}(\bar{k}m^3)$  {for each $k$},   {where $m$ is the dimension of the system as defined earlier}. 
\end{remark}

\begin{figure}[b!]
    \centering
    \includegraphics[width=6.6cm,height=4cm]{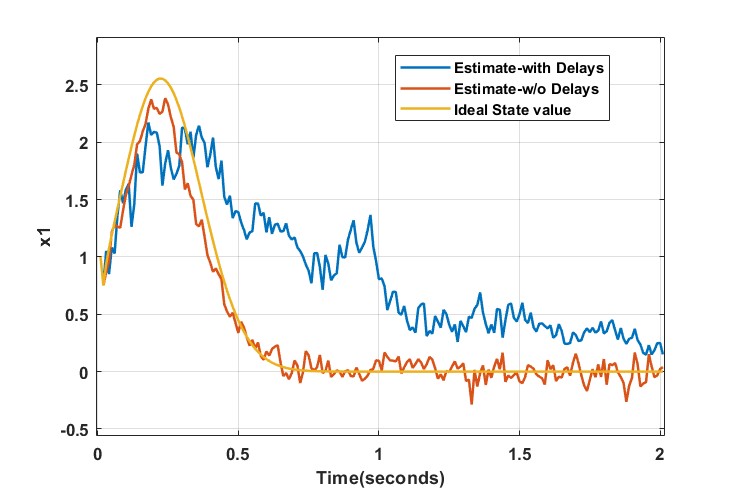}
    \caption{Estimates of state variable $x_1(k)$ with and without delays. Notice an MD higher than 100\% around 0.5 seconds compared to the ideal response.}
    \label{fig:comp}\vspace{-5mm}
\end{figure}
\begin{figure}[!t]
    \centering
    \begin{subfigure}[c]{0.49\linewidth}
    \includegraphics[width=4.3cm,height=3.1cm]{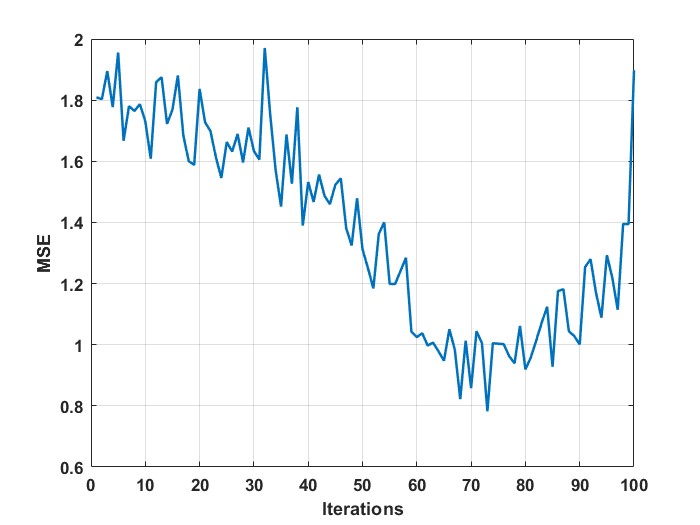}
    \end{subfigure}
    \begin{subfigure}[c]{0.49\linewidth}
    \includegraphics[width=4.3cm,height=3.1cm]{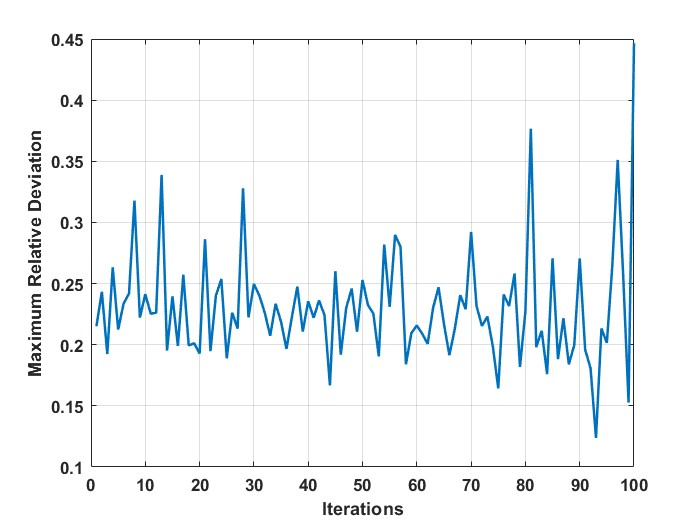}\vspace{-2mm}
    \end{subfigure}
    \caption{Implementing the greedy sensor subset selection (Algorithm 1): (a) Variation of MSE and (b) variation of MD on reducing the threshold. The values are ensemble means.}
    \label{fig:alg2}
\end{figure}

\vspace{-1 mm}
\section{Results and Discussions}\label{cha:results}
Consider the problem of DKF for the following LTV stochastic system modelled in discrete time as
\begin{align*}
    x(k+1)& = 
    \begin{bmatrix}
    0.5 & 0.25\\
    0.25 & 2^{-t_k}
    \end{bmatrix}
    x(k)  + w(k),\\
    z_i(k)&=H_ix(k)+v_i(k),
\end{align*}
with $H_i$ being either $\begin{bmatrix}
0 & 1
\end{bmatrix}$ or $\begin{bmatrix}
1 & 0
\end{bmatrix}$, $t_k\in [0,2]$, 
and the sampling time of the sensors is $T_s = 0.01$. We assume that there exists a transmission delay within fixed bounds from each filter node to the estimator. In this study, we consider a relatively large set of sensors by setting $n=2000$ sensor nodes. The initial condition of the system is set to $x_0=\begin{bmatrix}
1 & 1
\end{bmatrix}^{\top}$, and the system disturbance process follows $w_{k} \sim \mathcal{N}(0_{2 \times 1},0.1I_2), \forall k \geq 0$. The measurement noise for the sensor nodes is Gaussian distributed with zero mean and variances belonging to the range $[0,0.5]$ . Communication delays are assumed to be known, belonging to the range $[0,2]$ (in seconds). 
 The influence of delays on the state estimates tends to be more pronounced by considering more nodes, which may cause undesired deviations in the transient behaviour. An illustration is shown in Fig. \ref{fig:comp}. Moreover, one can observe deviations of the estimated states from their ideal counterparts in the steady-state region. This is due to higher measurement noise variances from several sensor nodes. 

The LTV system in this example, has only one time-varying element in $a_{22} = 2^{-t_k}$. Thus the structural matrices for each time step in finite time can be written as
\begin{equation*}
    \bar{A} = \begin{bmatrix}
    \star & \star\\
    \star & \star
    \end{bmatrix}
    \quad \bar{H_i} = \begin{cases}
    \begin{bmatrix}
    0 & \star\end{bmatrix},& \text{if } H_i = \begin{bmatrix}
    0 & 1\end{bmatrix}, \\
    \\
    \begin{bmatrix}
    \star & 0\end{bmatrix},& \text{if } H_i = \begin{bmatrix}
    1 & 0\end{bmatrix}.
\end{cases} 
\end{equation*}
Structural observability in~\cite{montanari}, therefore,  trivially holds for the described LTV system, since the states can be reached from one another and $a_{22} \to 0$ only if $t \to \infty$ while at least one of the states is being measured given $H_i$ being either $\begin{bmatrix}
0 & 1
\end{bmatrix}$ or $\begin{bmatrix}
1 & 0
\end{bmatrix}$. Thus, the considered system is structurally observable 
over $[t_0,t], \forall t_0,t$.

\begin{figure}[!b]
    \centering
    \includegraphics[width=6.6cm,height=4.3cm]{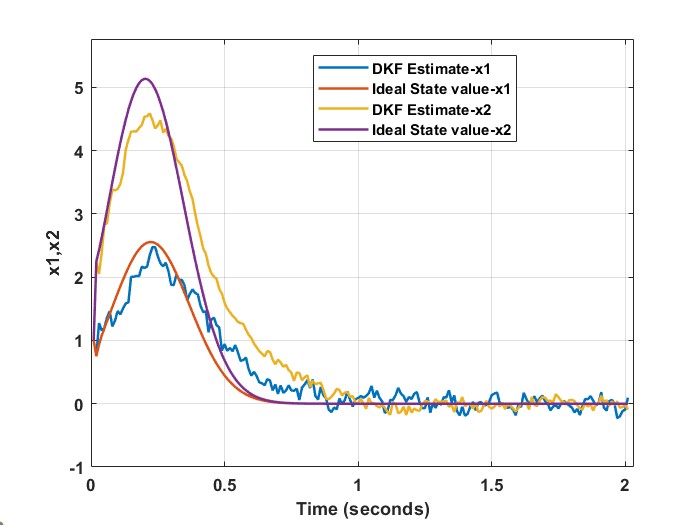}\vspace{-2mm}
    \caption{State estimates (Algorithm 1) corresponding to the lowest MSE from Fig.~\ref{fig:alg2}. The thresholds with the lowest MSE lead to the selection of 559 sensor nodes out of 2000.}
    \label{fig:alg2comp}
\end{figure}

Applying Algorithm 1, as expected, by reducing the threshold values for measurement noise and delay, we observe a decreasing trend in MSE, as can be seen in Fig.~\ref{fig:alg2}a. However, further reduction (after the $73^{rd}$ iteration) of the thresholds results in higher MSE. This is because by lowering the thresholds too low, the number of sensors satisfying the threshold conditions becomes less as well resulting in less accurate estimates. 
According to Fig.~\ref{fig:alg2}, there exists an optimal threshold where  {there is a trade-off between the} number of selected sensors and delays along with noise levels. This observation indeed confirms the results in~\cite{1188769} wherein the threshold-based policies where shown to produce optimal performance when thresholds are selected properly.

\begin{figure}[!t]
    \centering
    \includegraphics[width=6.6cm,height=4.3cm]{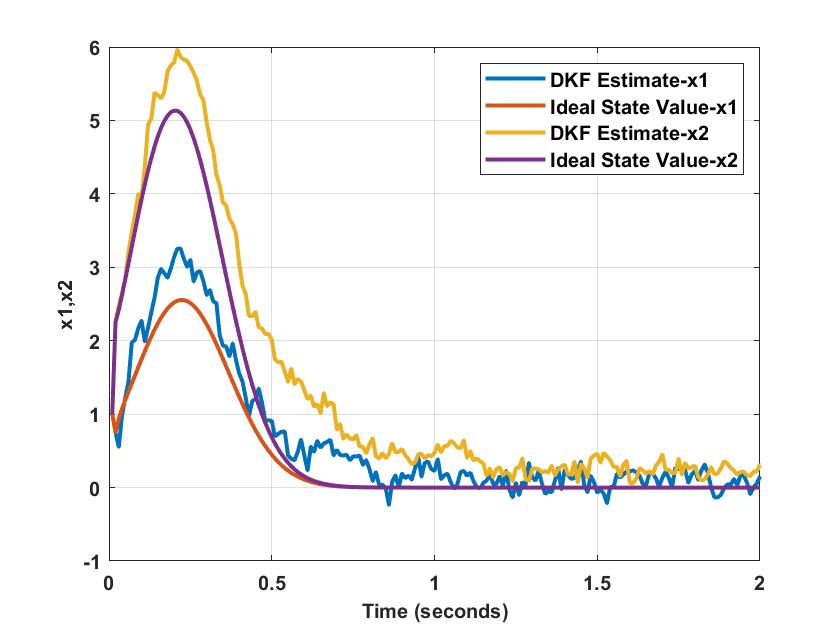}\vspace{-2mm}
    \caption{State estimates (Algorithm 2) with $n=2000$, 192 nodes selected, returning an MSE $=2.3871$, MD $=0.2593$.}
    \label{fig:alg3sim}
\end{figure}

While it appears that the magnitude of the local minima for MD decreases with reducing the threshold (Fig.~\ref{fig:alg2}b), unlike the case of MSE, one cannot observe a unique decreasing trend for MD. One similarity with respect to the MSE in Fig.~\ref{fig:alg2}a is that we observe a peak in MD at the end of all iterations owing to too few nodes for estimation.

  {Implementing Algorithm~\ref{alg:alg3} to select a subset of nodes that guarantee stability bounds, results in an estimate comparable to that of Algorithm~\ref{alg:alg1}. 
} When applying Algorithm~\ref{alg:alg3} and to ensure desired stability margin, $\hat{\beta}$ should satisfy (\ref{betahat}) when calculating $\widetilde{\mathcal{I}}(k)$. Furthermore, a reasonably small value of $\bar{k}$ has to be chosen to ensure that the provided stability specifications are achievable. The estimates obtained from Algorithm~\ref{alg:alg3} are illustrated in Fig.~\ref{fig:alg3sim}.

Further, Algorithm~\ref{alg:alg3} is also tested to select a subset of sensor nodes with stochastic delays. Each sensor node was subject to a constant and additive stochastic delay randomly selected from $\mathcal{N}(0,2T_s)$. Over a Monte Carlo study of $25$ runs, it is observed that the estimates calculated from the subset of nodes with stochastic delays were comparable to the estimates computed from the subset of nodes with constant delays. A slight decrease in the MSE compared to the constant delay case is observed, which can be attributed to the lower number of selected sensor nodes. The average MSE over the trials is around 1.96 with a variance of 0.03.
Moreover, the MD shows a slight increase compared to the constant delay case, with an average of 0.29 and a variance of 0.02. The number of nodes selected varied across the different runs, with an average of 22 nodes over the runs and a high variance of 14. This is expected since the number of nodes that provide a stable estimate depends on delays. Fig.~\ref{fig:alg3simrd} illustrates the simulation results of the estimates on the trials with the lowest MD. 

The performances using Algorithms 1 and 2 are similar in terms of estimating the state variables, as seen from Figs. \ref{fig:alg2comp} and \ref{fig:alg3sim}. 
\begin{figure}[!t]
    \centering
    \includegraphics[width=6.6cm,height=4.1cm]{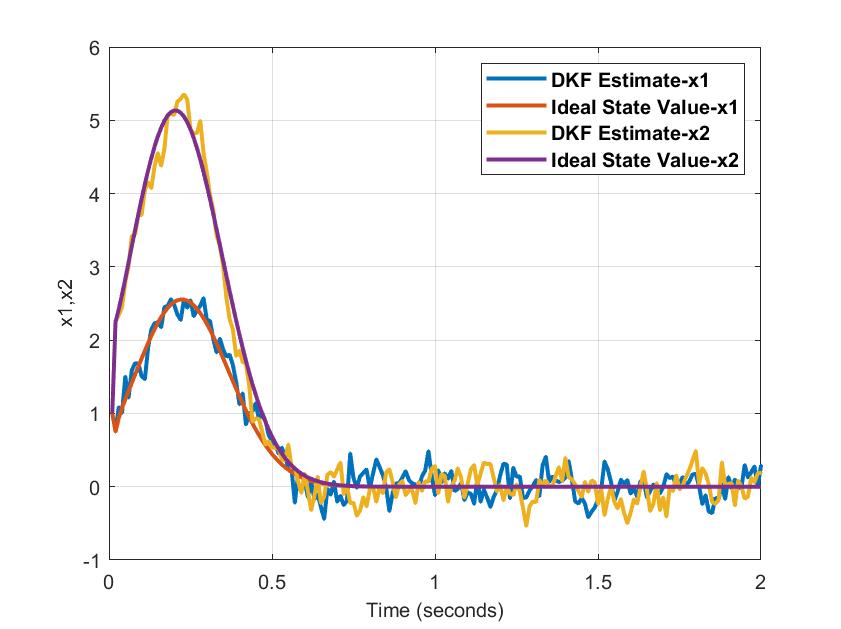}\vspace{-1mm}
    \caption{State estimates of $x_1(k)$ and $x_2(k)$ with $n=2000$ in the presence of stochastic delays. Algorithm 2, selected 20 nodes returns MSE $=1.92$ and MD $=0.11$.}
    \label{fig:alg3simrd}
\end{figure}
A noticeable advantage is that Algorithm~\ref{alg:alg3} results in the selection of sensor node subsets with far fewer nodes, compared to Algorithm~\ref{alg:alg1}.
\vspace{-1mm}
\section{Conclusions}\label{cha:conclusion}
We propose a subset selection algorithm to choose a set of sensor nodes that are deemed suitable for distributed Kalman Filter state estimation in the presence of stochastic communication delays   {for LTV Systems}. The algorithm is based on minimizing the number of sensor nodes in the subset while ensuring the stability of the state estimates. We then compare our algorithm with a greedy subset selection method that acts as a benchmark. Simulation results show comparable performance between our proposed approach and the greedy algorithm.   {Directions for future research include investigating stability-based node selection algorithms for DKF in partially observable LTV systems and extending the stability results under network-induced reliability and resource limitation phenomena.  
}

\bibliographystyle{unsrt}
\bibliography{root} 

\end{document}